\documentclass[12pt,a4paper]{amsart}
\usepackage[margin=2.5cm]{geometry}

\usepackage[utf8]{inputenc}
\usepackage{amsmath, amssymb, amsfonts,amsthm,amsopn,amscd}
\usepackage{dsfont}
\usepackage{graphicx}
\usepackage{longtable}
\usepackage{thm-restate}

\usepackage{mathtools}

\usepackage{upgreek}
\usepackage{xcolor}
\usepackage{nameref}
\usepackage[breaklinks]{hyperref}
\hypersetup{
    colorlinks,
    linkcolor={blue!50!black},
    citecolor={blue!50!black},
    urlcolor={blue!20!black}
}

\DeclareMathOperator{\tr}{tr}

\DeclareMathOperator{\wt}{wt}

\newcommand{\bra}[1]{\mathinner{\langle #1|}}
\newcommand{\ket}[1]{\mathinner{|#1\rangle}}

\newcommand{\dyad}[1]{| #1\rangle \langle #1|}
\newcommand{\ot}[0]{\otimes}
\newcommand{\one}[0]{\mathds{1}}
\renewcommand{\a}{\alpha}
\renewcommand{\b}{\beta}

\newcommand{\R}{\mathds{R}}
\newcommand{\C}{\mathds{C}}

\newcommand{\PP}{\mathcal{P}}

\newtheorem{theorem}{Theorem}
\newtheorem*{theorem*}{Theorem}

\newtheorem{observation}[theorem]{Observation}
\newtheorem{remark}[theorem]{Remark}

\newcommand{\nn}{\nonumber}

\begin{document}

\title[
On two maximally entangled couples
]{On two maximally entangled couples}

\date{\today}

\date{\today}

\author{Felix Huber$^{1}$}
\address{$^{1}$~Division of Quantum Computing,
Institute of Informatics,
Faculty of Mathematics, Physics and Informatics,
University of Gdańsk,
Wita Stwosza 57, 80-308 Gdańsk, Poland
}
\email{felix.huber@ug.edu.pl}

\author{Jens Siewert$^{2,3}$}
\address{$^{2}$~University of the Basque Country UPV/EHU,
Department of Physical Chemistry,
P.O.Box 644, 48080 Bilbao, Spain
}
\address{$^{3}$~Ikerbasque Foundation, Basque Foundation for Science, 48000 Bilbao, Spain
}
\email{jens.siewert@ehu.eus}

\thanks{
F.H.'s research was funded in whole or in part by the National Science Centre, Poland 2024/54/E/ST2/00451
and by the Polish National Agency for Academic Exchange under the Strategic Partnership Programme grant BNI/PST/2023/1/00013/U/00001.
For the purpose of Open Access,
the author has applied a CC-BY public copyright licence to any Author Accepted Manuscript (AAM) version arising from this submission.
J.S. was supported by Grant PID2021-126273NB-I00 funded by
MCIN/AEI/10.13039/501100011033 and by ”ERDF A
way of making Europe” as well as by the Basque Government through Grant No. IT1470-22.
The corresponding author is FH}

\begin{abstract}
 In a seminal article, Higuchi and Sudbery showed that a pure four-qubit state can not be maximally entangled across every bipartition.
 Such states are now known as absolutely maximally entangled (AME) states.
 Here we give a series of old and new proofs of the fact that no four-qubit AME state exists.
 These are based on invariant theory, methods from coding theory, and basic properties from linear algebra such as the Pauli commutation relations.
\end{abstract}

\maketitle

\section*{Introduction}
This article is dedicated to Ryszard Horodecki
on the occasion of his  80th birthday.
Over the past several decades, his pioneering contributions to the foundations
of quantum information theory--particularly in the areas
of entanglement
and the study of nonlocality, as exemplified in Refs.~\cite{Horodecki1995,HORODECKI19961,PhysRevLett.80.5239}--have profoundly shaped the field.
Together with his sons Paweł, Michał, and Karol Horodecki, his work has laid essential conceptual
and mathematical  groundwork that continues to inform and
inspire researchers across a wide spectrum of quantum information science.
Notably, his insights into
bound entanglement in highly mixed states—entangled yet undistillable—have revealed both the richness and the inherent limitations of quantum correlations.
In this spirit,
we turn our attention to a particular type of multipartite states: the so-called
absolutely maximally entangled states.
While distinct from Ryszard Horodecki's main areas
of inquiry, this work is driven by the same fundamental questions about the structure and implications of quantum correlations that continue to animate his contributions to the field.

\section*{Absolutely maximally entangeld states}

The existence of absolutely maximally entangled (AME) states, 
and in particular that of four-qubit AME states, 
is one of the simplest instances of the quantum marginal problem~\cite{1742-6596-36-1-014} and beautifully illustrates a
key difficulty at the core of all of many-body physics:
It asks whether there exists a pure $n$-qubit state (more generally, $n$-qudit state), 
such that every bipartition is maximally entangled;
or equivalently, for a pure $n$-qubit state such that all reduced density matrices~\footnote{Here we use the words 
{\em reduced density matrices} and {\em marginals}  interchangeably. 
A $k$-body reduced density matrix or $k$-body marginal is obtained after partially tracing out all but $k$ subsystems.}
of size $\lfloor \tfrac{n}{2} \rfloor$ are maximally mixed.
Naturally, in the case of two and three qubits, the Bell and Greenberger-Horne-Zeilinger states satisfy this requirement.
But what about four qubits?

In a seminal article (revisited in~\nameref{sec:HS}), Higuchi and Sudbery found that no such state can possibly exist. This opened up a field of research that dealt with both bounds and constructions on absolutely maximally entangled and $k$-uniform states,
the latter having all $k$-body marginals maximally mixed~\cite{PhysRevA.92.032316, PhysRevLett.118.200502, Karol36officers}. We point to the review~\cite{rajchelmieldzioc2025absolutelymaximallyentangledpure} for a comprehensive background on absolutely maximally entangled states, in particular regarding constructions. 

Let us sketch connections of this problem to quantum coding theory: if every vector in a subspace has maximally mixed $k$-body marginals, we are dealing with a pure quantum code of distance $k+1$~\cite{Huber_2020}.
Indeed it was early recognized by Scott that methods of coding theory are useful in the characterization of quantum entanglement~\cite{PhysRevA.69.052330}.
We present here two proofs inspired by coding theory.
\nameref{sec:shadow} relies on Rains' shadow inequalities~\cite{796376}, a monogamy relation constraining the purities in multipartite quantum states.
\nameref{sec:lovasz} rests on the Lov\'asz theta number,
a quantity traditionally used to bound the independence number and thus the capacity of a graph, but which has recently shown to be also applicable to quantum codes~\cite{munne2024sdpboundsquantumcodes}. 
This is a quite unexpected appearance of an invariant from graph theory, 
which here is used for a problem of continuous character, 
i.e. for the non-existence of a rank-$1$ subspace with some given properties.

A second influential connection is that to invariant theory: clearly the property of being maximally entangled does not depend on the local basis chosen, and thus it should be possible to characterize such states with local unitary polynomial invariants. \nameref{sec:GW}, as well as the newly presented proof using the $L+M+N= 0$ identity in ~\nameref{sec:LMN}
follow this line of thinking.
Perhaps Rains' article on polynomial invariants of quantum codes is the most complete combination of these two approaches~\cite{817508}; see also the subsequent book by Nebe, Rains, and Sloane~\cite{NebeRainsSloane2001}.

Finally, we present two proofs that are based on more pedestrian approaches: these proofs chiefly use the Schmidt decomposition and the commutation relations among Pauli matrices. The even-odd correlation approach of \nameref{sec:even_odd} was first used in Ref.~\cite{PhysRevLett.118.200502} to show the non-existence of a seven-qubit state, but surprisingly the same argument applies to all n-qubit states excluding $n=2,3,5,6$.
\nameref{sec:opposite} relies on the fact that all three-qubit GHZ-like states are local unitary invariant, and that the the correlations in a convex combination of two of such states must sufficiently cancel to allow for maximally mixed two-body marginals. 

These different approaches highlight that there often exists more structure in these problems 
than strictly necessary for a proof.
Furthermore, we see that entanglement theory, 
and in particular absolutely maximally entangled states,
manages to tie together very different subfield of mathematics.
While this collection of proofs shows that a lot is known, it is still difficult to pinpoint the essence of this non-existence fact. Although we are perfectly able to outline the problem and can elucidate its technical solution by connecting it to many different facts and mathematical properties of quantum many-body systems, we lack a simple intuition for how local features of the parts relate to (and constrain) the global structure of the system as a whole.

We now provide a collection of proofs for the following statement:
\begin{observation}
 Pure four-qubit states that are maximally entangled across every bipartition do not exist.
\end{observation}
In the remainder of the article we review seven proofs for this fact. 
To our knowledge, \nameref{sec:opposite} and \nameref{sec:LMN} are presented here for the first time.

\section*{Proof 1: Higuchi-Sudbery}
\label{sec:HS}
In their remarkable work, Higuchi and Sudbery~\cite{HiguchiSudbery2000} analyzed the possibility of maximum entanglement among four qubits.
While entanglement for two and three qubits was well understood at the time, for $n>3$ qubits not much was known. They recognized that the obvious $n$-qubit entangled Greenberger-Horne-Zeilinger (GHZ) state is probably not the ``most entangled'', where that concept  had yet to be specified.
They did the latter by requiring the maximum number of maximally mixed reductions -- from then on this has been the common definition for {\em absolute maximal entanglement} (AME)~\cite{helwig2013absolutelymaximallyentangledstates}.
They noticed that it is not obvious that such a state would exist: The Bell state ($n=2$) and the GHZ state ($n=3$) correspond to the definition. For four qubits there at least exists the state
\[
\ket{\Psi}\ =\ \frac{1}{2}
             \left( \ket{0000}+\ket{0111}+\ket{1001}+
                    \ket{1110}
                    \right)\ \ ,
\]
which has four (out of six possible) maximally mixed two-qubit reductions. Hence, Higuchi and Sudbery  studied the question whether there can exist a four-qubit state
with six maximally mixed reduced states of two qubits. In the following we briefly sketch how they found their negative answer.

Assuming a four-qubit AME state $\ket{\psi}$ exists,
we can expand it in the computational basis as
\begin{align}
    \ket{\psi}=\sum a_{jklm}\ket{jklm} \ \ .
\end{align}
The conditions that all two-qubit reductions be fully mixed translate to the fact that
\begin{align}
    a_{jklm}\ =\ \frac{1}{2}\left(U_1\right)_{jk,lm}
            \ =\ \frac{1}{2}\left(U_2\right)_{jl,km}
            \ =\ \frac{1}{2}\left(U_3\right)_{jm,kl}
            \ \ ,
\end{align}
where $U_1$, $U_2$ and $U_3$ are 
two-qubit unitary matrices.
By applying local unitaries  to
$\ket{\psi}$ it can be achieved that
\begin{align}
    a_{1000}\ =\ a_{0100}\ =\ 0\ =\ a_{0010}\ =\ a_{0001}\ \ .
\label{eq:rows-and-columns}
\end{align}
To see that this is possible consider only the first equality, $a_{1000}=a_{0100}=0$. The four components $a_{mn00}$, $m,n\in\{0,1\}$, may formally be viewed as a two-qubit vector. As is well known the Schmidt form of this vector can be obtained by applying local unitaries to the qubits, corresponding here to unitary operations on the first and second qubit of $\ket{\psi}$, which yield the first equality. Here we need to make sure that $|a_{0000}|$ {\em increases}, that is, we choose the local unitaries such that the new $a_{0000}$ equals the larger Schmidt coefficient. An analogous argument holds for the second part of the equality treated separately (by considering 
$a_{00pq}$) and is achieved by applying local unitaries to the third and fourth qubit of $\ket{\psi}$. Again we choose the unitaries such that $|a_{0000}|$ increases. By iterating this procedure, one achieves the maximal $a_{0000}$,
at which point both parts of the equality hold simultaneously.

Subsequently Higuchi and Sudbery argue by using the orthogonality
conditions of the columns and rows of the unitary matrices $U_j$ ($j=1,2,3$) that more coefficients of $\ket{\psi}$ must be zero. With the remaining non-zero elements it is impossible to simultaneously guarantee unitarity of
$U_1$, $U_2$ and $U_3$, which concludes their proof.

\section*{Proof 2: Even and odd correlations}
\label{sec:even_odd}

The following strategy was used in Ref.~\cite{PhysRevLett.118.200502}
to disprove the existence of a seven-qubit AME state,
but it also works - in simpler form - for the non-existence of a four-qubit AME.

The proof follows the strategy of:
1) recognizing that all marginals of an AME state are proportional to projectors, and thus satisfy a quadratic equation;
2) expanding the density matrix into correlations of different weights;
and
3) splitting the quadratic equation by an even-odd parity rule on the weights of anti-commutators of Pauli matrices
leading to a contradiction.

\begin{proof}
Suppose there exists a a four-qubit AME state $\ket{\phi}$.
Let $\varrho_{123} = \tr_4(\dyad{\phi})$ be its three-body reduced density matrix on the first three systems.
Since the state $\ket{\phi}$ is AME,
all of its two-body reduced states are maximally mixed,
$\tr_1(\varrho_{123}) = \tr_2(\varrho_{123})= \tr_3(\varrho_{123}) = (\one \ot \one) / 4$.
Therefore $\varrho_{123}$ can be written as
\begin{equation}\label{eq:3_RDM_correl}
 \varrho_{123} = \frac{1}{8}(\one + P_3)\,,
\end{equation}
where $P_3$ contains three-body terms only. That is, $P_3$ contains only tensor products
of three non-identity Pauli matrices
$\sigma_\a, \sigma_\b, \sigma_{\gamma}\neq \one$,
\begin{equation}
 P_3 = \sum_{\a,\,\b,\,\gamma \neq 0} \, c_{\a\b\gamma} \, \sigma_\a \ot \sigma_\b \ot \sigma_\gamma\,,
\end{equation}
so that $P_3$ vanishes under any partial trace applied to 
any single subsystem.

The single-body marginal $\varrho_4 = \tr_{123}(\dyad{\phi})$ is maximally mixed and has eigenvalues $(\tfrac{1}{2}, \tfrac{1}{2})$.
By the Schmidt decomposition of pure states,
it can be seen that the complementary reduction $\varrho_{123}$ has the same non-zero spectrum.
Consequentially, $\varrho_{123}$ is proportional to a projector and satisfies,
\begin{equation}\label{eq:3_RDM_proj}
 \varrho_{123}^2 = \frac{1}{2}\varrho_{123}\,.
\end{equation}
Expanded in terms of the correlations as in Eq.~\eqref{eq:3_RDM_correl}, we decompose this projector relation as
\begin{equation}\label{eq:3_RDM_expand}
 \big( \frac{1}{8}(\one + P_3) \big)^2 = \frac{1}{2}\cdot\frac{1}{8}(\one + P_3)\,.
\end{equation}
Simplifying, one obtains
\begin{equation}\label{eq:3_RDM_expand_simpli}
 P_3^2 = 3\one + 2P_3\,.
\end{equation}

We now recall the even-odd Lemma of Ref.~\cite{PhysRevLett.118.200502}.
For this, denote the weight of a Pauli operator as even (odd),
if it contains a non-identity Pauli matrix in an even (odd) number of positions.
Then in any anti-commutator of Pauli operators, the weights of correlations satisfy
\begin{align}
\{\text{even}, \text{even}\} &\rightarrow \text{even}\,, \nn\\
\{\text{odd} , \text{odd}\} &\rightarrow \text{even}\,, \nn\\
\{\text{even}, \text{odd}\} &\rightarrow \text{odd}\,.
\end{align}
We use this to split the rhs and lhs of Eq.~\eqref{eq:3_RDM_expand_simpli} into its even and odd parts,
\begin{align}
 P_3^2 = \frac{1}{2} \{P_3, P_3\}&= 3 \one \quad\quad \text{(even)} \label{eq:split_even}\\
 0 &= P_3           \quad\quad \text{(odd)}  \label{eq:split_odd}
\end{align}
Now the second equation yields clearly a contradiction: necessarily $P_3 \neq 0$ due to the fact that $\varrho_{123}$ cannot have full rank, which follows from the Schmidt decomposition of a pure state.
As a consequence, a four-qubit AME state does not exist.
\end{proof}

\begin{remark}
The proof for the non-existence of the seven-qubit AME state requires only one additional argument:
the Schmidt decomposition shows that the joint state is an eigenvector of every marginal,
\begin{equation}
 \tr_S(\varrho) \ot \one_{S^c} \ket{\psi} \propto \ket{\psi}\,.
\end{equation}
where $S^c$ denotes the complement of $S$ in $\{1, \dots,n\}$.
The even-odd Lemma is then used to split two projector relations as in Eq.~\eqref{eq:3_RDM_proj}
each into two contributions, for which the eigenvector relations provide a contradiction.

\end{remark}

\section*{Proof 3: Opposite two-body correlations}
\label{sec:opposite}
We need a short lemma, 
originally proven in Ref.~\cite{SCHLIENZ199639},
which we now derive through an iterated Schmidt decomposition.

\begin{restatable}{lemma}{ghzlemma}\label{lem-1}
All pure three-qubit states with three maximally mixed 1-RDM are local unitary equivalent to the standard
GHZ state $(\ket{000}+ \ket{111})/\sqrt{2}$.
\end{restatable}
The proof can be found in the Appendix~\ref{app:GHZ_proof}.

We recall the definition of quantum error-correcting codes.
A code with parameters $(\!(n,K,\delta)\!)_2$ is able to encode a $K$-dimensional system into $n$ qubits,
such that all errors acting on at most $\lfloor \tfrac{\delta-1}{2}\rfloor$ subsystems can be corrected.
Then $n$ is block-length, $K$ the dimension, and $\delta$ the distance of the code.
More precisely, a subspace of $(\C^2)^{\ot n}$ is a $(\!(n,K,\delta)\!)_2$ code,
if and only if the Knill-Laflamme conditions hold
for all Pauli strings $E$ with weight $\text{wt}(E) < \delta$,
\begin{equation}\label{eq:KLF}
 \Pi E \Pi = c_{E} \Pi\,, \quad\quad\quad c_E \in \C\,,
\end{equation}
where $\Pi$ is the projector associated with the code subspace.
Here $\operatorname{wt}(E)$ is the number of subsystems a tensor-product operator $E$ acts non-trivially on.

A code is said to be {\em pure} 
if it additionally satisfies $c_E = \tr(E)/2^n$.
Thus for a pure code, $\tr_{S^c}(\Pi)$ is maximally mixed for all $S \subseteq \{1,\dots, n\}$ of size $|S| <\delta$ and where $S^c$ is the complement of $S$ in  $\{1,\dots, n\}$.
 This can be seen by writing
 $E = E_S \ot \one_{S^c}$ and taking the trace over Eq.~\eqref{eq:KLF},
\begin{equation}
\tr[\Pi E \Pi]
=  \tr[\Pi(E_S \ot \one_{S^c})]
=  \tr[\tr_{S^c}(\Pi) E_S]
= \tr(E) / 2^n \tr(\Pi)\,,
  \end{equation}
which vanishes for all $E\neq \one$.
Thus $\tr_{S^c}(\Pi)$ is maximally mixed for all $|S|<\delta$.

It is known that if there exists a pure code with parameters $(\!(n,K,d)\!)_D$, then there also exists a pure $(\!(n-1,DK,d-1)\!)_D$
code~\cite[Theorem 19]{681316}. This is known as {\em code propagation}, 
and the new code (sub-)space is obtained by taking a partial trace over a single subsystem of the old code space.
After the partial trace only $n-1$ subsystems remain and thus the new block-length is $n-1$, 
while it can be shown that this procedure also 
decreases the distance by one with the dimension of the code space increasing by a factor of the local dimension 
that has been traced out.

\smallskip
We now show the non-existence of four-qubit AME states.
\begin{proof}
Having all two-body marginals maximally mixed, a four-qubit AME state corresponds to a pure quantum error-correcting code with parameters $ (\!(4,1,3)\!)_2$\cite{681315,PhysRevA.69.052330}.
Propagating the $(\!(4,1,3)\!)_2$ by
a partial trace on the last party yields a pure code with parameters $ (\!(3,2,2)\!)_2$,
represented by the associated density matrix $\varrho_{ABC}$.
By the Schmidt decomposition, this reduced state
has eigenvalues $(\tfrac{1}{2}, \tfrac{1}{2})$ and thus fulfills the relation $\varrho_{ABC}^2 = \frac{1}{2} \varrho_{ABC}$.
We can write this state as a convex combination,
\begin{equation}
 \varrho_{ABC} = \frac{1}{2} \big(\dyad{v} + \dyad{w}\big),
\end{equation}
where $\ket{v}$ and $\ket{w}$ are orthogonal vectors spanning the code space.
The code is pure with distance two,
and thus
every vector in its span must be 1-uniform, including
the vectors $\ket{v}$ and $\ket{w}$.
We can thus expand them as
\begin{align}\label{eq:GHZ_ansatz}
 \dyad{v} &= \frac{1}{8}(\one + P_2 + P_3)\,, \nn\\
 \dyad{w} &= \frac{1}{8}(\one + Q_2 + Q_3)\,,
\end{align}
where $P_2, Q_2$ contain only Pauli operators of weight two,
and $P_3,Q_3$ contain only Pauli operators of weight three.
In this way, $\dyad{v}$ and $\dyad{w}$
have vanishing one-body correlations.

Lemma~\ref{lem-1}
shows that any pure state 1-uniform state
must be locally unitary equivalent to the three-qubit Greenberger-Horne-Zeilinger state.
We can thus restrict $\ket{v}$ and $\ket{w}$ to be of GHZ type. In particular,
choose a local basis so that $\ket{v} = \ket{\text{GHZ}}$.

Let us analyze the spectra of their correlations. The state
$\ket{\text{GHZ}} = \frac{1}{\sqrt{2}}( \ket{000} + \ket{111})$ decomposes in terms of Pauli operators as
\begin{equation}
    \dyad{\text{GHZ}} = \frac{1}{8} (III
    + ZZI + ZIZ + IZZ - XYY - YXY - YYX + ZZZ)\,.
\end{equation}
Now group the two and three-body correlations,
\begin{align}
    P_2 &= ZZI + ZIZ + IZZ\,, \nn\\
    P_3 &= - XYY - YXY - YYX + ZZZ \,.
\end{align}
It can be checked that the eigenvalues are
\begin{align}\label{eq:evals}
\sigma(P_2) &= (-1, -1, -1, -1, -1, -1,  3,  3)\,, \nn\\
 \sigma(P_3) &= (-4,  0,  0,  0,  0,  0,  0,  4)\,.
 \end{align}
Due to Lemma~\ref{lem-1}, all three-qubit GHZ-type  states are local unitary equivalent. As a consequence, also the eigenvalue structures of their corresponding two- and three-body correlations are equal.
That is, $\sigma(Q_2) = \sigma(P_2)$ and
$\sigma(Q_3) = \sigma(P_3)$.

Now note that because $\varrho_{ABC} = (\dyad{v} + \dyad{w})/2$ is two-uniform (being the reduction of a two-uniform state), we must necessarily have that $Q_2 = - P_2$ in Eq.~\eqref{eq:GHZ_ansatz}.
Thus for $ \dyad{w}$, we need to make the Ansatz
\begin{equation}
    \dyad{w} = \frac{1}{8}(\one - P_2 + Q_3)\,,
\end{equation}
in the hope that $Q_3$ renders the matrix
$\dyad{w}$ positive semidefinite.
However, the spectrum obtained from the two-body correlations $P_2$ of $\ket{v}$ alone already has negative eigenvalues,
\begin{equation}
 \sigma(\one - P_2) = (-2, -2,  2,  2,  2,  2,  2,  2) \,.
\end{equation}

By Eq.~\eqref{eq:evals}, $Q_3$ has the non-zero eigenvalues $\{4,-4\}$.
Now recall that that
 $\lambda_{\max}(A+B) \leq \lambda_{\max}(A) + \lambda_{\max}(B)$ for all hermitian matrices $A,B$.
Then $\ket{w}$ being pure requires that $\lambda_{\text{max}}(Q_3) \geq 6$ so that
$\lambda_{\text{max}}(\one - P_2 + Q_3) = 8$.
This leads to a contradiction, as by Eq.~\eqref{eq:evals} one has $\lambda_\text{max}(Q_3) = 4$. This ends the proof.
\end{proof}

\section*{Proof 4: Polynomial Invariants of Four Qubits}
\label{sec:LMN}

Polynomial invariants are polynomial expressions of the coefficients of a pure multi-party quantum state (with local dimension $d$), which remain invariant under the application of local SL($d$,$C$) operations. These operations are represented by complex $d\times d$ matrices with determinant equal to 1. 
Polynomial invariants are relevant for the quantification of entanglement because they allow for a straightforward construction of measures for different classes of genuine multipartite entanglement~\cite{Verstraete2003,EltschkaSiewertReview2014}. 

For states of four qubits it is possible to construct the full generating set of invariants under the group SL$(2,C)^{\otimes 4}$ as was demonstrated by Luque and Thibon~\cite{LT2003}. 
It consists of four elements, one degree-2 invariant, two degree-4 invariants and one degree-6 invariant.
By writing a four-qubit state in the computational basis,
$\ket{\psi}=\sum a_{j_1 j_2 j_3 j_4}\ket{j_1 j_2 j_3 j_4}$
(where $j_k\in\{0,1\}$) one obtains three degree-4 invariants~\cite{LT2003}
\begin{subequations}
\begin{align}
          L\ \equiv\ \det{\mathsf{L}}\ =\ \left|
                \begin{array}{cccc}
                a_{0000} & a_{0100} & a_{1000} & a_{1100}\\
                a_{0001} & a_{0101} & a_{1001} & a_{1101}\\
                a_{0010} & a_{0110} & a_{1010} & a_{1110}\\
                a_{0011} & a_{0111} & a_{1011} & a_{1111}
                \end{array}
                \right|\ \ ,
\label{eq:LTL}
\\
          M\ \equiv\ \det{\mathsf{M}}\ =\ \left|
                \begin{array}{cccc}
                a_{0000} & a_{1000} & a_{0010} & a_{1010}\\
                a_{0001} & a_{1001} & a_{0011} & a_{1011}\\
                a_{0100} & a_{1100} & a_{0110} & a_{1110}\\
                a_{0101} & a_{1101} & a_{0111} & a_{1111}
                \end{array}
                \right|\ \ ,
\label{eq:LTM}
\\
          N\ \equiv\ \det{\mathsf{N}}\ =\ \left|
                \begin{array}{cccc}
                a_{0000} & a_{0001} & a_{1000} & a_{1001}\\
                a_{0010} & a_{0011} & a_{1010} & a_{1011}\\
                a_{0100} & a_{0101} & a_{1100} & a_{1101}\\
                a_{0110} & a_{0111} & a_{1110} & a_{1111}
                \end{array}
                \right|\ \ ,
\label{eq:LTN}
\end{align}
\end{subequations}
where any two of them are linearly independent.
It can be seen by inspection (i.e. evaluating the determinants) that they fulfill the relation
\begin{align}
    L\ +\ M\ +\ N\ =\ 0\ \ .
\label{eq:sumLMN}
\end{align}
The properties of these invariants give rise to yet another approach to prove the non-existence of the four-qubit AME state.

\begin{proof}
In Ref.~\cite{Eltschka2012}  it was observed that the matrices $\mathsf{L}$, $\mathsf{M}$ and $\mathsf{N}$ are closely related to the two-qubit reduced states of $\ket{\psi}$, namely,
\begin{subequations}
\begin{align}
     \rho_{12}\ =\ & \mathsf{L}^T \mathsf{L}^*\ \ , \label{eq:L_a}
     \\
     \rho_{24}\ =\ & \mathsf{M}\mathsf{M}^{\dagger}\ \ ,
     \\
     \rho_{14}\ =\ & \mathsf{N}^T\mathsf{N}^*\ \ . \label{eq:L_c}
\end{align}
\end{subequations}
For an AME state we must have
$\rho_{12}=\rho_{24}=\rho_{14}=\frac{1}{4} \one_4$. 
After taking a complex conjugate for 
Eqs.~\eqref{eq:L_a} and \eqref{eq:L_c}, 
these equalities imply that the matrices $\mathsf{L}$, $\mathsf{M}$, and $\mathsf{N}$, up to a factor $\frac{1}{2}$, must be unitary. Accordingly, the columns and rows of each matrix form an orthogonal basis, with all vectors normalized to $\frac{1}{2}$.
Moreover, we can choose the global phase of $\ket{\psi}$ so that, because of Eq.~\eqref{eq:sumLMN},
\begin{align}
    L\ =\ \frac{1}{16}\ , \ \
    M,N\in
    \Bigg\{ \frac{1}{16}
            \mathrm{e}^{\frac{2\pi\mathrm{i}}{3}},
            \frac{1}{16}
            \mathrm{e}^{\frac{4\pi\mathrm{i}}{3}}
           \Bigg\}\ \ .
\label{eq:phasesLMN}
\end{align}
Consider now, e.g., the two-qubit column vectors of $\mathsf{L}$. We may apply local unitaries to the qubits of the first column (which amounts to applying local unitaries to the third and fourth qubit of $\ket{\psi}$) so as to obtain the Schmidt decomposition, i.e.,
\[
           a_{0001}\ =\ a_{0010}\ =\ 0\ \ .
\]
Similarly we can apply local unitaries to the first and second qubit of $\ket{\psi}$ and get the Schmidt form also for the first row vector of $\mathsf{L}$,
\[
           a_{0100}\ =\ a_{1000}\ =\ 0\ \ .
\]
Using an argument and procedure analogous to that of \nameref{sec:HS} in deriving Eq.~\eqref{eq:rows-and-columns}, we obtain that both equalities are satisfied simultaneously.
Assume for the moment that $a_{0011}\neq 0$ (we discuss the opposite case below). The orthogonality conditions then require the following structure of $\mathsf{L}$ (up to sign changes and column or row permutations),
\begin{align}
    \mathsf{L}\ =\ \left[
                \begin{array}{cccc}
                s & 0 & 0 & -t\\
                0 & u & -v^* & 0\\
                0 & v &  u^* & 0\\
                t & 0 & 0 & s
                \end{array}
                \right]\ \ ,
\end{align}
where $s,t\in \mathds{R}$ and $u,v\in \mathds{C}$. This fixes also the entries of $\mathsf{M}$ and $\mathsf{N}$.
We then find,
\begin{align}
         M\ =\ (s^2-|u|^2)(t^2-|v|^2)\ ,
\ \ \    N\ =\ -(t^2+|u|^2)(s^2+|v|^2)\ \ .
\label{eq:ameMN}
\end{align}
It is evident from Eq.~\eqref{eq:ameMN} that
$M,N\in \mathds{R}$, in contradiction with
Eq.~\eqref{eq:phasesLMN}.

If, instead, $a_{0011}=0$ it follows from the normalization for all rows and columns in $\mathsf{L}$, $\mathsf{M}$,
$\mathsf{N}$ that
\[
    a_{1100}=a_{0101}=a_{1010}=a_{1001}=a_{0110}=0\ \ .
\]
This, however, leads to $L=M=N=0$ which again is a contradiction to Eq.~\eqref{eq:phasesLMN}.
\end{proof}

We note that in its spirit, this proof version is rather similar to the original Higuchi-Sudbery idea.
The coefficients of $\ket{\psi}$ may be arranged into three different two-qubit unitaries, which are linked to one another by partial transposition or reshuffling of the matrix elements~\cite{Bengtsson_Zyczkowski_2017,Karol36officers}. By exploiting the properties of the coefficients of $\ket{\psi}$ under local unitaries it can be shown that not all of the three matrices can simultaneously
be unitary.

\section*{Proof 5: Shadow inequalities}
\label{sec:shadow}

This proof uses the shadow inequality by Rains~\cite{
796376, 817508}:
Denote by $\varrho_S = \tr_{S^c} (\varrho)$ the marginal of $\varrho$ on subsystem $S$,
where $S^c$ denotes the complement of $S$ in $\{1,\dots,n\}$.
The inequality states that
for all states $\varrho$ and subsets $T\subseteq \{1,\dots n\}$,
it holds that~\footnote{
This trace inequality has been promoted to an {\em operator} inequality in Ref.~\cite{PhysRevA.98.052317}.}
\begin{equation}\label{eq:shadow_ineq}
S_T =  \sum_{S \subseteq \{1\dots n\}} (-1)^{|S \cap T|}\tr_S(\varrho_S^2) \geq 0\,.
\end{equation}
The non-negativity of this expression can be understood by projecting $\varrho^{\ot 2} $ locally
onto symmetric and anti-symmetric subspaces given by $\tfrac{1}{2}(\one + \operatorname{swap})$ and $\tfrac{1}{2}(\one - \operatorname{swap})$ and taking the trace~\cite{681316}.
Expanding the resulting expression using the swap trick, 
$\tr[\operatorname{swap} (A \otimes B)] = \tr(AB)$,
yields the inequality in Eq.~\eqref{eq:shadow_ineq}.
The shadow inequality can be understood as a monogamy of entanglement constraint which governs the sharing of quantum correlations~\cite{PhysRevA.98.052317}.

\begin{proof}
The marginals of AME states are maximally mixed, and by the Schmidt decomposition all of its purities are known.
For a four-qubit AME state,
\begin{align}
 \tr(\varrho_i^2) &= \tr(\varrho_{ijk}^2) =\frac{1}{2} \,,\nn\\
 \tr(\varrho_{ij}^2) &=\frac{1}{4}\,.
\end{align}
We evaluate Eq.~\eqref{eq:shadow_ineq} with $T = \{A,B,C,D\}$ and get
\begin{equation}
 S_{ABCD} = 1
 - (4\cdot\tfrac{1}{2})
 + (6\cdot \tfrac{1}{4} )
 - (4\cdot \tfrac{1}{4})
 + 1
 = -\tfrac{1}{2}\,.
\end{equation}
However, the shadow inequality requires that $S_{ABCD}$ be non-negative, leading to a contradiction.
As a consequence, a four-qubit AME state does not exist.
\end{proof}

\section*{Proof 6: Gour-Wallach}
\label{sec:GW}
Independently of Rains, Gour and Wallach~\cite{GourWallach2010} presented a proof that also relies on an inequality involving invariants, however, without explicit reference to Rains'
shadow inequalities.

\begin{proof}

Define the linear entropy
$      \tau_{X|\bar{X}}\ =\ 2\left(1-\tr (\rho_X^2) \right)
$,
and let
\begin{align}
    \tau_1 &=
\frac{1}{4}\big(\tau_{A|BCD} + \tau_{B|ACD} + \tau_{C|ACD} + \tau_{D|ABC}\big)\,, \nn\\
    \tau_2 &=
\frac{1}{3}\big(\tau_{AB|BC} + \tau_{AC|BD} + \tau_{AD|BC}\big)\, .
\end{align}
Gour and Wallach showed that for a four-qubit state~\cite{10.1063/1.2435088},
\begin{equation}\label{eq:GW}
    4 \tau_1 - 3\tau_2 = \tau_{ABCD} \geq 0\,,
\end{equation}
where $\tau_{ABCD} = |\!\bra{\psi} \sigma_y^{\ot 4} \ket{\psi^*}\!|^2$ is non-negative due to it being an absolute value squared.
However, a state with maximally mixed two-qubit marginals must have
$\tau_1 = 1$ and $\tau_2 = 3/2$, which is in contradiction to inequality~\eqref{eq:GW}.
\end{proof}
\begin{remark}
Note that for qubits,
$\tau_{ABCD}$ is nothing else than $S_{ABCD}$ from Eq.~\eqref{eq:shadow_ineq}.
\end{remark}

\section*{Proof 7: Lov\'asz bound}
\label{sec:lovasz}

This proof follows the strategy of Ref.~\cite[Corollary 9]{munne2024sdpboundsquantumcodes},
which provides a complete semidefinite programming hierarchy for the existence of qubit quantum codes with given parameters $(\!(n,K,d)\!)$.
Here, $n$ denotes the number of qubits the code contains, $K$ is the dimension of the code space, and $d$ the distance.
Here, we only need a relaxation of this hierarchy, which is given by the Lov\'asz number $\vartheta$.

Let $G$ be a graph of $N$ vertices,
where we write $i \sim j$ if $\{i,j\}$ is an edge in the graph.
The Lov\'asz  number $\vartheta$ is given by the following semidefinite program~\cite{GALLI2017159}:
\begin{align}\label{eq:lovasz}
	\vartheta(G) \quad=\quad \text{maximize} \quad & \sum^N_{i=1} M_{ii}  \nn \\
	\text{subject to} \quad & M_{ii} = a_i \quad \forall i \in V\,,   \nn\\
	&M_{ij} = 0  \quad \,\,\text{if} \quad i \sim j\,, \nn \\
	&\Delta = \begin{pmatrix} 1 & a^T\\ a & M
\end{pmatrix} \succeq 0 \,.
\end{align}
Here both $a \in \R^N$ and $M \in \R^{N\times N}$ are real.
It is known that the Lovász number~$\vartheta$ bounds the independence number $\alpha$ of a graph,
that is the maximum size of a set in which no two vertices are connected,
as $\a(G)\leq \vartheta(G)$.

\begin{proof}
Our aim is to construct a anti-commutativity graph $G$ that captures the correlations of the AME state,
and to show that its theta number is too small for it to correspond to a pure state.
Denote by $\sigma_x,\sigma_y,\sigma_z$ the three Pauli matrices,
and by $\sigma_0$ the identity matrix.
Let $\PP_n$ be the $n$-qubit Pauli basis, that is the set of the $4^n$ elements of the form
$
E_\a = \sigma_{\a_1} \ot \dots \ot \sigma_{\a_n}
$,
where $\sigma_{\a_i} \in \{\sigma_0, \sigma_x, \sigma_y, \sigma_z\}$.
With some abuse of notation, map the multi-index $\a$ to an integer between $0$ and $4^n-1$, so that $E_0 = \one$.

Suppose there exists a four-qubit AME state $\varrho$.
Let $\langle E_\a \rangle = \tr( E_\a \varrho)$ be its expectation value on Pauli operator $E_\a$.
Then one can define a moment matrix with entries
\begin{equation}
    \Gamma_{\a\b} = \langle E_\a^\dag \rangle  \langle E_\b \rangle \langle E_\a^\dag  E_\b \rangle\,,
\end{equation}
where $0 \leq \a,\b\leq 4^n-1$.
It can be seen that $\Gamma$ is positive semidefinite:
for any $v\in \C^{4^n}$ one has,
\begin{align}
 v^\dag \Gamma v &=
 \sum_{\a,\b} v_\a^*
 \tr\big(
 \langle E_\a^\dag \rangle \langle E_\b \rangle E_\a^\dag  E_\b \varrho
 \big)
 v_\b  \nn\\
 &=
 \tr\Big[ \big( \sum_{\a} v_\a \langle E_\a \rangle E_\a \big)^\dag
 \big( \sum_{\b} v_\b   \langle E_\b \rangle E_\b \big) \varrho \Big]\geq 0\,,
\end{align}
since it is the product between a hermitian square and $\varrho$, both of which are positive semidefinite operators.

Observe that
$\Gamma_{00} = \tr(\langle \one \rangle \one \varrho) = 1$ and
\begin{align}
 \Gamma_{\a0}
 &= \tr\big( \langle E_\a \rangle \langle \one \rangle  E_\a \one \varrho\big) \nn\\
 &= \tr\big( \langle E_\a \rangle \langle E_\a  \rangle E_\a E_\a\varrho\big) = \Gamma_{\a\a}\,,
\end{align}
due to the fact that the Pauli operators are unitary and hermitian, so that $E_\a^2 = \one$~\footnote{
Note that all Pauli operators are hermitian so we can safely drop any daggers.}.
Therefore the matrix~$\Gamma$ satisfies:
$\Gamma_{00} = 1$ and $\Gamma_{\a0} = \Gamma_{\a\a}$ for all $\a$.
Furthermore, from the normalization
$\tr(\varrho^2)=1$ it follows that
\begin{equation}\label{eq:Gamma_sum}
\sum_{\a=0}^{4^4-1} \Gamma_{\a\a} = \sum_{\a=0}^{4^4-1}  \langle E_\a \rangle^2 = 2^4\,,
\end{equation}
and from the maximally mixed marginals condition of AME states one has
\begin{align} \label{eq:gamma_entries1}
 \Gamma_{\a\b} &= 0 \quad\quad \text{if}
 \quad 0< \wt(E_\a E_\b) <d\,,\quad \nn\\
 \Gamma_{\a\a} & = 0 \quad \quad \text{if}
  \quad 0< \wt(E_\a) <d\,,
 \end{align}
where the weight $\text{wt}$ of a Pauli string is the number of subsystems it acts non-trivially on.
We observe that the second condition also implies that if
$0< \wt(E_\a) <d$ then also $\Gamma_{\a\b}$ for all~$\b$.
This is due to the fact that if a positive semidefinite
matrix has a diagonal entry equal to zero,
then all entries in the same row or column are zero too.
Note that one could also obtain this constraint directly from the definition of $\Gamma_{\a\b}$.

Now note that the transposed moment matrix satisfies the same properties. We can thus consider
$
    \hat \Gamma = \frac{1}{2}(\Gamma + \Gamma^T)\,,
$
for which additionally
 \begin{equation}\label{eq:gamma_entries2}
     \hat \Gamma_{\a\b} = 0 \quad\quad \text{if}\quad  E_\a E_\b + E_\b E_\a = 0\,,
 \end{equation}
 due to the commutation relations among the Pauli matrices.

Let the {\em weight} $\wt(E)$ be the number of subsystems a tensor-product operator $E$ acts non-trivially on.
In light of the conditions derived in Eqs.~\eqref{eq:gamma_entries1} and \eqref{eq:gamma_entries2},
 let us write
\begin{align}
 \a &\sim \b \quad\quad \text{if}\quad
 0<\wt(E_\a E_\b) < d \quad \text{or} \quad
 E_\a E_\b + E_\b E_\a = 0 \,, \nn\\
 \a &\sim \a \quad\quad \text{if} \quad
 0<\wt(E_\a) < d\,.
 \end{align}
Then if a four-qubit AME state exists, the following semidefinite program must,
due to Eq.~\eqref{eq:Gamma_sum},
reach a value of at least $2^4-1$:
\begin{align}\label{eq:AME_Lovasz}
\text{maximize}   \quad & \sum_{\a = 1}^{4^4-1} \hat \Gamma_{\a\a}\nn\\
\text{subject to} \quad
            & \hat \Gamma_{00} = 1 \nn\\
            &\hat \Gamma_{\a\a} = \Gamma_{\a0} \quad\quad \,\,\forall \a \in V\,,   \nn\\
            &\hat \Gamma_{\a\b} = 0 \quad\quad\quad\,\, \text{if} \quad \a \sim \b\,, \nn\\
            & \hat \Gamma \succeq 0 \,.
\end{align}
But Eq.~\eqref{eq:AME_Lovasz} is nothing else than the
Lov\'asz theta number [Eq.~\eqref{eq:lovasz}]
for the anti-commutativity graph given by the $n$-qubit Pauli basis
with the identity matrix removed, $\PP_n \backslash \{\one\}$, and additional edges among low-weight products.
Consequently, for an AME state to exist, we need to have $\vartheta(G) + 1 \geq 16$.
However, for the anti-commutativity graph of the four-qubit Pauli group $\PP_4 \backslash \{\one\}$,
$\vartheta(G) + 1 = 8$.
As a consequence, a four-qubit AME state does not exist.
\end{proof}

\begin{remark}
A similar strategy can be used to obtain SDP bounds on quantum codes: One constructs such a Lovasz bound, but then symmetrizes the resulting semidefinite program with the quaternary Terwilliger algebra~\cite{munne2024sdpboundsquantumcodes}.
\end{remark}

\section*{Discussion}
Finally, we should mention a proof that we skipped:
The proof by Grassl, Beth, and Pellizzari shows that no $(\!(3,2,3)\!)$ exists. As a consequence of, also no $(\!(4,1,3)\!)$.
This proof is rather intricate~\cite[Theorem 5]{PhysRevA.56.33}, relying on a sequence of arguments involving the Knill-Laflamme conditions, and the fact~\cite[Lemma 3]{PhysRevA.56.33} that for every two–dimensional subspace of $\C^2 \otimes \C^2$ there exists a basis that contains at least one product state. This leads to the conclusion that a $(\!(3,2,3)\!)$ code can be factored into a tensor product,
which contradicts that there does not exist a non-trivial code of length two.

In summary, we have discussed a variety of proofs for the nonexistence
of the four-qubit AME state. Our work highlights that this
nonexistence can be traced back to several well-established results
from diverse areas of quantum information theory. However, none of the
existing methods offers a simple explanation for why three completely
mixed two-qubit marginals cannot arise from a global pure four-qubit state. 
Nevertheless, we hope that our findings contribute to a deeper understanding of this issue by illuminating the interplay
between the various constraints and structural requirements.
This perspective may serve as a foundation for future work
aimed at uncovering more general principles governing multipartite
entanglement and marginal compatibility.

For future work, it will be interesting to see 
how much the semidefinite programming methods of Ref.~\cite{munne2024sdpboundsquantumcodes} will improve on the current bounds on quantum codes. 
We have used this method in shortened form in \nameref{sec:lovasz},
but the full machinery has already proven the non-existence of $(\!(8,9,3)\!)_2$ and $(\!(10,5,4)\!)_2$ codes, 
both of which are not detected by the linear programming bounds~\cite{681316, 1751-8121-51-17-175301}.
A key aspect in this work is the state polynomial optimization framework, and we except that other non-linear problems can also be approached by this type of machinery~\cite{PhysRevLett.132.200202, Huber_2024}. 
In these of applications, a key issue is typically the numerical precision of the infeasibility certificates involved. However, 
recent methods allow to extract rational infeasibility certificates, rendering these methods as good as any analytical proof~\cite{Leijenhorst_2024, PEYRL2008269}.

In the spirit of Refs.~\cite{Cadney_2014, 6134666}
Another interesting research direction would be to develop an analytical or numerical machinery to systematically find new entropy and rank inequalities.

\appendix
\section{Three-qubit GHZ states}
\label{app:GHZ_proof}
We would like to prove the following statement:

\ghzlemma*
\begin{proof}
The proof is carried out by explicit construction of the local unitaries required to transform a given state to the standard GHZ form.
Consider a three-qubit state $\ket{\psi_{ABC}}$  with three maximally mixed one-body reduced density matrices. We start with a Schmidt decomposition with respect to the bipartite cut that contains qubit $A$,
\begin{align}
    \ket{\psi_{ABC}}\ =\ \frac{1}{\sqrt{2}}
    \big( \ket{0}\ket{u_0}+\ket{1}\ket{u_1} \big)\ ,
\end{align}
with two orthogonal two-qubit states $\ket{u_0}_{BC}$, $\ket{u_1}_{BC}$. We can do another Schmidt decomposition
of $\ket{u_0}$,
which then reads,
\begin{align}
    \ket{u_0}\ =\ \frac{1}{\sqrt{2}}\left( \sqrt{\lambda}\ket{00}+
                      \sqrt{1-\lambda}\ket{11}\right)
                      \ \ \ ,\ \ \lambda\in \mathbb{R}\ .
\end{align}
For simpler notation we have used the local Schmidt bases of qubits
$A$, $B$ and $C$ as the new local computational bases.
Now, taking into account orthogonality of the two-qubit Schmidt vectors $\ket{u_0}$
and $\ket{u_1}$ we can write the latter in the local Schmidt
bases of $\ket{u_0}$,
\begin{align}
    \ket{u_1}\ =\ \frac{1}{\sqrt{1+|a|^2+|b|^2}}
                   \left( -\sqrt{1-\lambda}\ket{00}
                   +a\ket{01}+b\ket{10}+
                   \sqrt{\lambda}\ket{11}
                   \right)\ ,
\end{align}
where $a$ and $b$ are complex numbers.
This state can be substituted in $\ket{\psi_{ABC}}$.
If we now compute the reduced states $\rho_B$ and $\rho_C$ we find that the off-diagonal elements of $\rho_B$  - up to complex conjugation -
are proportional to
\begin{equation}\label{eq:offdiag1}
      -\sqrt{\lambda}\ a^*\ +\ b\ \sqrt{1-\lambda}
      \ ,
\end{equation}
whereas those of $\rho_C$ are proportional to
\begin{equation}\label{eq:offdiag2}
      -a\ \sqrt{1-\lambda} \ +\ b^*\ \sqrt{\lambda}
      \ .
\end{equation}
As the one-body reduced density matrices must be maximally mixed, this implies that
\begin{align}\label{eq:ratio}
    \frac{\sqrt{\lambda}}{\sqrt{1-\lambda}} \ =\
    \frac{b}{a^*} \ =\ \frac{a}{b^*}\ \ \ \ \ \ \text{for}\ \ a,b\neq 0\ .
\end{align}
To continue the proof, we must distinguish between two cases, (i) $
\lambda\neq\frac{1}{2}$ and (ii) $\lambda=\frac{1}{2}$.
{\color{blue} (i) If} $\lambda\neq\frac{1}{2}$ and $a, b\neq 0$ vanishing off-diagonal elements of the one-body reduced density matrices in Eq.~\eqref{eq:offdiag1} and \eqref{eq:offdiag2} is not possible. Therefore it must hold that
$a=b=0$.
Hence for $\lambda\neq\frac{1}{2}$
we obtain for $\psi_{ABC}$,
\begin{align}
    \ket{\psi_{ABC}}\ =\ \frac{1}{\sqrt{2}}
    \left(\sqrt{\lambda}\ket{000}+\sqrt{1-\lambda}\ket{011}-\sqrt{1-\lambda}\ket{100}+\sqrt{\lambda}\ket{111}\right)\ .
\label{eq:almostGHZ}
\end{align}
From here the application of
\begin{align}
     \left[ \begin{array}{cc}
                \sqrt{\lambda} & -\sqrt{1-\lambda}
        \\  \sqrt{1-\lambda} & \sqrt{\lambda}
                    \end{array}\right]_A\ \otimes
                   \ \one_B\ \otimes\ \one_C
\label{eq:LU3}
\end{align}
leads us to the standard GHZ state.

(ii) If instead $\lambda=\frac{1}{2}$ we have, according to Eq.~\eqref{eq:ratio}, $b=a^*$ so that 
$\ket{u_1}$ becomes
\[
\ket{u_1}\ =\ \frac{1}{\sqrt{1+2|a|^2}}
                   \left( -\frac{1}{\sqrt{2}}\ket{00}
                   +a\ket{01}+a^*\ket{10}+
                   \frac{1}{\sqrt{2}}\ket{11}
                   \right)\ .
\]
We reparametrize this state more conveniently as
\begin{align}
    \ket{u_1}\ =\ \frac{1}{\sqrt{2}}
    \left( -\sin{2\beta}\ket{00} +
    \text{e}^{i\theta}        \cos{2\beta}\ket{01} +
    \text{e}^{-i\theta}
            \cos{2\beta}\ket{10} +
            \sin{2\beta}\ket{11}
    \right)\ ,
\end{align}
so that the three-qubit state reads
\begin{align}
    \ket{\psi_{ABC}}_{\lambda=\frac{1}{2}}
    \ =\ \frac{1}{2}
    \left(\ket{000}+\ket{011}-\right.
    & \sin{2\beta}\ket{100}+
    \nonumber\\
    & \left.
        + \text{e}^{i\theta} \cos{2\beta}\ket{101}+\text{e}^{-i\theta}\cos{2\beta}\ket{110}+\sin{2\beta}\ket{111}\right)\ .
\label{eq:almostGHZ2}
\end{align}
Local unitary equivalence
 with the standard GHZ state is then confirmed
 by applying
\begin{align}
    H^{\otimes 3}\ \cdot\ \left(\one_A\ \otimes
     \left[ \begin{array}{cc}
        \text{e}^{-i\theta}  \cos{\beta} & \sin{\beta}
 \\   - \text{e}^{-i\theta}\sin{\beta} & \cos{\beta}
                    \end{array}\right]_B\
                    \otimes
                   \
    \left[ \begin{array}{cc}
        \text{e}^{i\theta}  \cos{\beta} & \sin{\beta}
 \\   - \text{e}^{i\theta}\sin{\beta} & \cos{\beta}
                    \end{array}\right]_C
                    \right)
\label{eq:LU3c}
\end{align}
to $\ket{\psi_{ABC}}_{\lambda=\frac{1}{2}}$, where
$H$ denotes the standard single-qubit Hadamard operation
\[
    H\ =\ \frac{1}{\sqrt{2}}
            \left[ \begin{array}{cc}
        1 & 1
 \\     1 & -1
                    \end{array}\right] \ .
\]
This ends the proof.
\end{proof}

\bibliographystyle{alpha}
\bibliography{current_bib}
\end{document}